\documentclass[journal]{IEEEtran}
\usepackage{enumitem}
\usepackage{cite, xcolor}
\usepackage[hidelinks]{hyperref}
\usepackage{microtype}
\usepackage{graphicx}
\usepackage{booktabs} % for professional tables
\usepackage{amsmath,amssymb}
\usepackage{enumerate}
\usepackage[linesnumbered,ruled,vlined]{algorithm2e}
\usepackage{algorithmic}
\usepackage{mathrsfs}
\usepackage{dsfont}
\usepackage{bbm}
\usepackage{float}
\usepackage{xcolor}
\usepackage{booktabs}
\usepackage{cite}
\usepackage{amsthm} % for environment proof

\usepackage{multirow}
\usepackage{pifont}% http://ctan.org/pkg/pifont

\usepackage{array, makecell} %

\usepackage{nicefrac}       % compact symbols for 1/2, etc.

\usepackage[obeyspaces]{xurl}

\usepackage[capitalise]{cleveref}
\Crefname{equation}{eq.}{eqs.}
\newtheorem{theorem}{Theorem}
\newtheorem{lemma}{Lemma}

\newtheorem{definition}{Definition}

\definecolor{dartmouthgreen}{rgb}{0.05, 0.5, 0.06}
\definecolor{bittersweet}{rgb}{1.0, 0.44, 0.37}
\definecolor{brightmaroon}{rgb}{0.76, 0.13, 0.28}
\definecolor{bluegray}{rgb}{0.4, 0.6, 0.8}
\definecolor{brandeisblue}{rgb}{0.0, 0.44, 1.0}
\definecolor{battleshipgrey}{rgb}{0.52, 0.52, 0.51}
\definecolor{cadetgrey}{rgb}{0.57, 0.64, 0.69}
\definecolor{lightgray}{rgb}{0.83, 0.83, 0.83}

\def\A{\mathbf{A}}

\def\x{\mathbf{x}}
\def\z{\mathbf{z}}

\def\h{\boldsymbol{h}}

\def\J{r}

\def\dimx{n}
\def\dimz{m}
\def\noise{\mathbf{e}}

\def\ex{\hat{\mathbf{x}}}

\def\eR{\boldsymbol{R}}

\def\thetaEnc{\theta_{\text{Enc}}}
\def\thetaDec{\theta_{\text{Dec}}}

\def\mask{\mathbf{b}}

\usepackage{amsmath,amsfonts,bm}

\def\eqref#1{equation~\ref{#1}}
\def\1{\bm{1}}

\DeclareMathAlphabet{\mathsfit}{\encodingdefault}{\sfdefault}{m}{sl}
\SetMathAlphabet{\mathsfit}{bold}{\encodingdefault}{\sfdefault}{bx}{n}

\begin{document}

\title{
%Structure-Aware Stealthy False Data Injection Attacks in AC Power Systems
% Limits of Detecting Physically Consistent False Data Injection Attacks on the Measurement Manifold
Limits of Residual-Based Detection for Physically Consistent False Data Injection
}

\author{
Chenhan Xiao,~\IEEEmembership{Student Member,~IEEE}
and~Yang~Weng,~\IEEEmembership{Senior Member,~IEEE}
\vspace{-1em}
\thanks{
C. Xiao and Y. Weng are with the School of Electrical, Computer and Energy Engineering at Arizona State University, \{cxiao20 and yang.weng\}@asu.edu.}
}

% make the title area
\maketitle

\begin{abstract}
False data injection attacks (FDIAs) pose a persistent challenge to AC power system state estimation. In current practice, detection relies primarily on topology-aware residual-based tests that assume malicious measurements can be distinguished from normal operation through physical inconsistency reflected in abnormal residual behavior.
This paper shows that this assumption does not always hold: when FDIA scenarios produce manipulated measurements that remain on the measurement manifold induced by AC power flow relations and measurement redundancy, residual-based detectors may fail to distinguish them from nominal data.
The resulting detectability limitation is a property of the measurement manifold itself and does not depend on the attacker's detailed knowledge of the physical system model. 
To make this limitation observable in practice, we present a data-driven constructive mechanism that incorporates the generic functional structure of AC power flow to generate physically consistent, manifold-constrained perturbations, providing a concrete witness of how residual-based detectors can be bypassed.
Numerical studies on multiple AC test systems characterize the conditions under which detection becomes challenging and illustrate its failure modes. The results highlight fundamental limits of residual-based detection in AC state estimation and motivate the need for complementary defenses beyond measurement consistency tests.

\end{abstract}


%%------------Introduction--------------%%
\section{Introduction}

The increasing digitization of power system monitoring and control has exposed state estimation to a growing class of cyber threats \cite{mantravadi2020securing,stouffer2022guide}. Among these, false data injection attacks (FDIAs) pose a particularly serious risk, as they manipulate measurement data to corrupt state estimates while evading conventional detection mechanisms \cite{liu2011false}. State estimation underpins a wide range of operational functions, including dispatch, contingency analysis, and protection coordination. Undetected errors in estimated states can therefore propagate to downstream control and operational decisions, potentially compromising system reliability and causing significant economic losses \cite{8686241,8468098}. In practice, FDIAs may arise through compromised field devices or communication and data aggregation infrastructure \cite{liu2011false,wang2020detection,hug2012vulnerability,zhang2018limited}. Real-world cyber incidents such as Stuxnet attack \cite{economist2010cyber} and Triton malware \cite{di2018triton} further illustrate the feasibility and impact of cyber intrusions on power-system infrastructures.

In current operational practice, detection of abnormal measurements is primarily embedded within topology-aware state estimation through residual-based detection mechanisms, which evaluate measurement consistency via residual errors relative to the physical measurement model. Among these, the statistical bad data detector (BDD) based on the chi-squared test serves as the canonical and most widely deployed instance \cite{liu2011false,liu2025survey}. More broadly, residual-based detection has been extended to learning-based mechanisms, motivated by increasing sensing density \cite{liu2025survey} and the need to implicitly learn measurement consistency when system topology or parameters are uncertain. Representative examples include autoencoder-based detectors \cite{musleh2022lstmAE,zideh2023PIConvAE}, subspace- and projection-based detection methods \cite{yu2015blind,wang2019cpad}, and generative-model based detectors \cite{liu2025survey}, which assess consistency through reconstruction or projection errors relative to an explicitly modeled or implicitly learned measurement mapping. Compared with time-domain anomaly detection methods, which rely on temporal changes rather than physical consistency and require sequential observations, residual-based detectors remain the primary detection paradigm in practical power grid operations owing to their tight integration with physical network models and state estimation pipelines \cite{liu2011false}. Most residual-based detectors share a common underlying assumption: malicious measurements induce noticeable deviations from nominal system behavior that manifest as abnormal residual behavior under the physical measurement model.

However, this paper shows that this assumption is not guaranteed. When FDIAs produce manipulated measurements that remain sufficiently close to the measurement manifold induced by power flow relations and measurement redundancy, residual-based detectors can become ineffective. This phenomenon reflects a fundamental detectability limit that is intrinsic to the geometry of the measurement manifold, rather than to any specific residual-based detection algorithm. Evidence of this limitation has appeared implicitly across a broad body of FDIA-related research. In model-based FDIA settings, physically consistent false measurements have been shown to evade residual-based BDD when detailed system model information is available to the attacker \cite{liu2011false}. To relax this knowledge assumption, a class of model-free FDIA methods has been developed that rely solely on historical measurement data to implicitly learn the underlying measurement structure. Representative approaches include independent component analysis \cite{esmalifalak2015stealthy}, principal component analysis \cite{yu2015blind}, low-rank matrix approximation \cite{yang2023false}, and matrix reconstruction techniques based on eigenvalue decomposition \cite{yang2022blind}, which are primarily effective in linear DC settings. To better handle the nonconvex nature of AC power flow relations, more recent studies have explored autoencoder- and generative-adversarial-network-based methods \cite{musleh2022attack,costilla2022attack,jiao2021new,wang2019cyber}. While these FDIA studies further illustrate the vulnerability of residual-based detection mechanisms, they typically exploit proximity to the measurement manifold in a heuristic manner and do not characterize the conditions under which residual-based detectors lose detectability. As a result, the detectability limit itself remains largely unformalized.

To formalize the detectability limit of residual-based detection, this paper revisits FDIA detectability from a geometric perspective. We show that detectability limit is a property of the measurement manifold: when manipulated measurements preserve physical consistency and remain on the measurement manifold, distinguishing them from nominal measurements becomes fundamentally difficult for residual-based detection schemes. Here, the measurement manifold refers to the low-dimensional structure induced in the measurement space by physical laws and measurement redundancy \cite{wang2019cyber}.

To make this detectability limit observable in a systematic and reproducible manner, we introduce a data-driven constructive mechanism that reconstructs the measurement manifold from historical data and generates physically consistent perturbations constrained to it. This construction is used solely as a diagnostic tool to expose and illustrate the detectability limit, rather than as its underlying cause. A natural starting point for such a construction is the autoencoder, which is widely used to learn low-dimensional representations of measurement data. However, as we show, standard autoencoder architectures do not guarantee manifold-consistent perturbations in nonlinear AC systems. This motivates a customized, physics-guided design that explicitly incorporates the functional structure of AC power flow. In particular, we leverage a symbolic basis transformation that lifts the system state into a higher-dimensional representation in which the measurement mapping becomes linear, while encoding dominant nonlinear dependencies, such as trigonometric and bilinear interactions, directly into the representation. As a result, the decoder operates on a space that already reflects the physical structure of the measurement mapping. This construction enables formal guarantees regarding when residual-based detectors lose detectability, in contrast to prior learning-based approaches that rely purely on data and offer no such assurances.

For illustrating the detectability limit in practice, we conduct experiments on IEEE 14-, 30-, 39-, 57-, 118-, and 200-bus systems, covering a wide range of network sizes and operating conditions. Rather than benchmarking attack performance, the experiments are designed to expose failure modes of residual-based detection under different system and data conditions. Historical measurement data are generated using MATPOWER \cite{zimmerman2010matpower}, with realistic load variations incorporated through real-world power consumption profiles. This use of historical measurement data is consistent with prior studies on model-free FDIA analysis and reflects contemporary power system data-handling practices, where measurement data are increasingly processed by external analytics platforms \cite{bidgely, verdigris}, introducing potential exposure during data transmission and storage \cite{banik2024survey}. We further examine the sensitivity of the observed detectability limits to the amount of available historical data.

The remainder of this paper is organized as follows. Section \ref{sec:background} reviews the FDIA problem in AC power systems and examines key factors governing residual-based detection and detectability. Section \ref{sec:attack} introduces a constructive, data-driven mechanism used to reveal the detectability limits associated with physically consistent measurements constrained to the measurement manifold. Section \ref{sec:experiments} presents numerical studies on representative AC power system test cases to illustrate the detectability limit in practice. Section~\ref{sec:conclusion} concludes the paper.

% #############  Preliminary  #############
\section{Measurement Geometry and Limits of Residual-Based Detection}
\label{sec:background}

This section reviews AC state estimation and residual-based detection from a measurement-geometry perspective, emphasizing how physical consistency and manifold structure determine the fundamental limits of detectability. While the classical chi-squared BDD is used as a representative analytical instance, the geometric insights and detectability limits developed in this section apply more broadly to residual-based detection schemes that assess measurement consistency via residual or reconstruction errors.

\vspace{-0.5em}
\subsection{AC State Estimation and Measurement Consistency}

In power system state estimation, the control center collects a set of measurements $\z=(z_1,\cdots,z_{\dimz})\in \mathbb{R}^{\dimz}$ and estimates the system states $\x=(x_1,\cdots,x_{\dimx})\in \mathbb{C}^{\dimx}$, which typically consists of voltage magnitudes and phase angles. Under the AC measurement model, the measurements are related to the system state by
\(
\z = \h(\x) + \noise,
\)
where $\h(\cdot)$ denotes the nonlinear measurement function determined by physical power flow relationships, and $\noise\in\mathbb{R}^{\dimz}$ represents measurement noise. While the specific realization of $\h(\cdot)$ depends on network topology and line parameters, its functional form follows standard physical laws shared across AC power systems. The noise $\noise$ is commonly modeled as a Gaussian, $\noise \sim \mathcal{N}(\boldsymbol{0},\eR)$, with covariance matrix $\eR=\mathrm{diag}(\sigma_1^2,\cdots,\sigma_{\dimz}^2)$.

Given the measurements $\z$, the state estimator computes an estimate $\ex$ by solving the weighted least-squares problem \cite{abur2004power}:
\(
\ex = \arg\min_{\x} (\z-\h(\x))^\top \eR^{-1}(\z-\h(\x)).
\)
Accurate state estimation is critical for many operational tasks, including economic dispatch \cite{rahman2014impact} and contingency analysis \cite{paul2022modified}. Consequently, the ability to detect abnormal measurements before they corrupt the estimated states is of central importance to secure system operation.

\vspace{-0.5em}
\subsection{Bad Data Detection and a Geometric Interpretation}

BDD is used to evaluate the consistency between observed measurements and the estimated system state \cite{abur2004power,wang2020detection}. The associated normalized residual error is defined as
\begin{align}\label{eq:residual}
\J(\z) 
% &:= \min_{\x}\sum_{i=1}^{\dimz}(z_i-\h_i(\x))^2/\sigma_i^2 \\
&=\min_{\x}(\z-\h(\x))^\top \eR^{-1}(\z-\h(\x)). 
\end{align}
Under the assumption that $\z$ contains no abnormal data and the noise is Gaussian, the statistic $\J(\z)$ approximately follows a Chi-squared distribution with $\dimz-\dimx$ degrees of freedom \cite{abur2004power}. Accordingly, measurements are flagged as abnormal if $\J(\z)$ exceeds a threshold $\tau=\chi^2_{(\dimz-\dimx),1-\alpha}$ for a prescribed significance level $\alpha$. Although derived here for BDD, the residual $\J(\z)$ corresponds to the squared distance from the measurement $\z$ to the measurement manifold and therefore captures the same consistency criterion evaluated by a broad class of residual-based detectors.

Beyond its statistical interpretation, the residual error admits a useful geometric meaning. When measurement variances are identical, $\sigma_i^2\equiv\sigma^2$, the quantity $\J(\z)$ is proportional to the squared Euclidean distance from $\z$ to the set
\vspace{-0.5em}
\begin{equation}\label{eq:manifold}
\mathcal{H}:=\{\h(\x)\mid\x\in\mathbb{C}^{\dimx}\}.
\vspace{-0.5em}
\end{equation}
This interpretation extends to heterogeneous noise levels through appropriate rescaling of the measurement space. The set $\mathcal{H}$ represents all physically valid measurements consistent with the AC model and is referred to throughout the paper as the \emph{measurement manifold}. Due to measurement redundancy and structured nonlinear dependencies among state variables, $\mathcal{H}$ forms a low-dimensional nonlinear manifold embedded in the ambient measurement space $\mathbb{R}^{\dimz}$. This geometric viewpoint is illustrated in the left panel of Fig.~\ref{fig:manifold}. It highlights that residual-based detection fundamentally evaluates the distance between the observed measurements and the measurement manifold, rather than detecting specific attack mechanisms.

\begin{figure}[H]
    \centering
    \vskip -0.1in
    \includegraphics[width=1\linewidth]{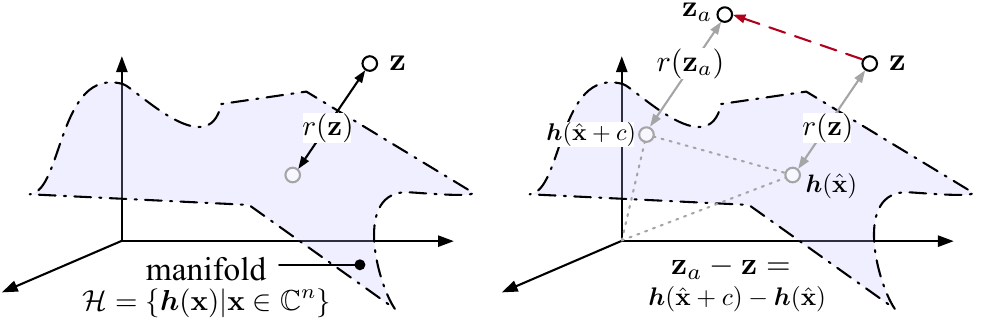}
    \vskip -0.05in
    \caption{{\bf Left}: The residual error $\J(\z)$ in BDD corresponds to the squared distance between the measurement $\z$ and the measurement manifold $\mathcal{H}$. {\bf Right}: Model-based perturbations modify measurements along the manifold $\mathcal{H}$, resulting in unchanged residual error.}
    \label{fig:manifold}
    \vskip -0.1in
\end{figure}

\vspace{-1em}
\subsection{Detectability Limits Induced by the Measurement Manifold}

The geometric interpretation above reveals a fundamental property of residual-based detection: abnormal measurements that move off the measurement manifold tend to create large residuals and are therefore detectable. In contrast, perturbations that remain close to the manifold may be difficult to distinguish from nominal operation. This observation motivates the following detectability limit, which characterizes when abnormal measurements are likely to bypass residual-based detection.

\begin{definition}[Manifold-Induced Detectability Limit]\label{def:stealth-condition}
Let the measurement manifold be $\mathcal{H}=\{\h(\x)\mid\x\in\mathbb{C}^{\dimx}\}$, and suppose BDD declares a measurement vector $\z$ abnormal whenever $\J(\z)\ge\tau=\chi^2_{(\dimz-\dimx),1-\alpha}$. For a mapping $g:\mathbb{R}^{\dimz}\to\mathbb{R}^{\dimz}$ that perturbs measurements, a necessary condition for the perturbed measurements to remain undetected under BDD is
\begin{equation}
\mathbb{P}_{\z\sim\mathcal{Z}}\!\left(\min_{\hat{\z}\in\mathcal{H}}\|g(\z)-\hat{\z}\|_2^2<\tau\right)\ge1-\alpha,
\end{equation}
where $\mathcal{Z}$ denotes the nominal measurement distribution.
\end{definition}

This detectability limit does not depend on how perturbation $g(\cdot)$ is generated. Rather, it shows an intrinsic limitation of residual-based detection: measurements that remain sufficiently close to the measurement manifold cannot be reliably distinguished from normal data based on residual magnitude alone.
Classical model-based FDIA constructions \cite{liu2011false,hug2012vulnerability} satisfy Definition~\ref{def:stealth-condition}. When full knowledge of the measurement model $\h(\cdot)$ is available, perturbations can be constructed by shifting measurements along the manifold, which leaves the residual error unchanged, as illustrated in the right-panel of Fig.~\ref{fig:manifold}. These constructions represent a special case of the more general geometric condition in Definition~\ref{def:stealth-condition}  and demonstrate that detectability limit arises from manifold geometry rather than from any specific attack algorithm.

In practice, however, access to the exact system model $\h(\cdot)$ is limited, particularly for external adversaries \cite{liu2011false}. This motivates a broader question addressed in the remainder of the paper: whether physically consistent, manifold-aligned perturbations can arise without explicit knowledge of $\h(\cdot)$, and what this implies for the fundamental limits of residual-based detection. Section~\ref{sec:attack} addresses this question by introducing a constructive mechanism used to reveal the detectability limits associated with the measurement manifold.

\section{A Constructive Mechanism for Revealing Detectability Limits}
\label{sec:attack}

Section~\ref{sec:background} showed that residual-based detection primarily tests whether measurements lie sufficiently close to the measurement manifold $\mathcal{H}$, with the chi-squared BDD serving as a canonical analytical instance. This implies a detectability limitation: any perturbation mechanism that produces measurements that remain within the detection acceptance region, as characterized in Definition~\ref{def:stealth-condition}, can be difficult to distinguish from nominal data using residual magnitude alone. The goal of this section is to introduce a constructive mechanism that makes this limitation observable using only historical measurements and generic knowledge of AC measurement structure. The mechanism is not the reason the limitation exists. It is used to demonstrate how manifold-aligned perturbations can emerge without explicit knowledge of the system-specific model $\h(\cdot)$.

From Definition~\ref{def:stealth-condition}, remaining undetected by residual-based detection requires that perturbed measurements stay close to $\mathcal{H}$. When the system model $\h(\cdot)$ is unknown, the manifold structure must be inferred implicitly from historical measurement data. A standard approach for learning low dimensional structure from high dimensional observations is the autoencoder \cite{goodfellow2016deep}. An autoencoder consists of an encoder that maps measurements to a latent representation and a decoder that reconstructs measurements from that representation. Letting $\mathrm{Enc}(\cdot)$ and $\mathrm{Dec}(\cdot)$ denote the encoder and decoder, respectively, the autoencoder is trained to minimize the reconstruction loss
\begin{equation}\label{eq:AE-standard}
\min_{\thetaEnc, \thetaDec} 
\mathbb{E}_{\z}\left\| \z - \text{Dec}\left(\text{Enc}\left(\z\right)\right)\right\|^2,
\end{equation}
where $\thetaEnc$ and $\thetaDec$ are network weights.

In DC systems, this architecture reduces to a linear model in which the encoder and decoder are linear maps. In that case, the autoencoder is equivalent to principal component analysis \cite{baldi1989neural}, which has been successfully applied to model free constructions in DC settings and admits statistical guarantees for subspace recovery \cite{yu2015blind,xiao2025}. In AC systems, however, linear representations are inadequate due to the structured nonlinearity of $\h(\cdot)$ \cite{costilla2020combining,jin2018power}. In particular, $\h(\cdot)$ contains trigonometric and bilinear interactions that are difficult to learn reliably from data without incorporating structural constraints \cite{rudy2017data}. As a result, even when a standard autoencoder achieves small reconstruction error on nominal data, its decoder behavior under latent perturbations is generally unconstrained. This means that a perturbation applied in latent space can map to a measurement that deviates substantially from the true measurement manifold, thereby failing to satisfy manifold-induced detectability limit in Definition~\ref{def:stealth-condition}. The following lemma formalizes this limitation.

\begin{lemma}[Standard Autoencoder Does Not Ensure Manifold-Aligned Perturbations]\label{lemma:standard-auencoder}
Consider a standard autoencoder $(\mathrm{Enc},\mathrm{Dec})$ trained by Eq.~(\ref{eq:AE-standard})
on noise-free measurement data $\z \in \mathcal{H}$ and achieving
exact reconstruction, i.e., $\mathrm{Dec}(\mathrm{Enc}(z))=z$ for all $z\in\mathcal{H}$. In general, this does not imply that the perturbations generated by
\(
    g(\z) = \text{Dec}(\text{Enc}(\z)+\mathbf{c})
\)
satisfies the manifold-induced detectability limit in
Definition~\ref{def:stealth-condition}, where $\mathbf{c}$ is a nonzero latent perturbation.
\end{lemma}

\begin{proof}
Let $\mathcal{L}:=\{\mathrm{Enc}(\z)\mid \z\in\mathcal{H}\}\subset\mathbb{R}^{d}$ denote the image of the measurement manifold $\mathcal{H}$ under the encoder. The exact-reconstruction assumption $\mathrm{Dec}(\mathrm{Enc}(\z))=\z$ for all $\z\in\mathcal{H}$ constrains the decoder only on $\mathcal{L}$, and imposes no restriction on $\mathrm{Dec}(u)$ for $u\notin\mathcal{L}$. Because $\mathcal{H}$ is a low-dimensional manifold embedded in $\mathbb{R}^{m}$, its encoded image $\mathcal{L}$ is generally a proper subset of $\mathbb{R}^{d}$.
Consequently, for any $\z\in\mathcal{H}$, there exist perturbations $\mathbf{c}$ such that
\(
\mathrm{Enc}(\z)+\mathbf{c}\notin\mathcal{L}.
\)
For such perturbed latent representations, the decoder output
\(
\mathrm{Dec}(\mathrm{Enc}(\z)+\mathbf{c})
\)
is unconstrained by the reconstruction objective and can lie arbitrarily far from the true measurement manifold $\mathcal{H}$.
Therefore, the resulting attacked measurement can violate the manifold-induced detectability limit in Definition~\ref{def:stealth-condition}.
\end{proof}

The failure mode in Lemma~\ref{lemma:standard-auencoder} can be illustrated by a simple one-dimensional example.
Let the measurement manifold be $\mathcal{H}=\{x^2\mid x\in\mathbb{R}\}$, and define the encoder and decoder as
$\mathrm{Enc}(z)=z^{1/3}$ and $\mathrm{Dec}(u)=u^3$, respectively.
This encoder-decoder pair achieves exact reconstruction on $\mathcal{H}$.
However, for $z\in\mathcal{H}$, a negative perturbation $c$ such that $\mathrm{Enc}(z)+c<0$ yields
$\mathrm{Dec}(\mathrm{Enc}(z)+c)<0\notin\mathcal{H}$.
Thus, perfect reconstruction can coexist with arbitrarily poor off-manifold behavior under small latent perturbations.

These observations indicate that the limitation of standard autoencoders is not from insufficient training or model capacity, but from the lack of structural constraints on how latent perturbations are mapped to measurement space.
To obtain a controlled perturbation mechanism that remains aligned with physically consistent structure, it is necessary to constrain the decoder so that perturbations map to a physically consistent structured set in measurement space.
This motivates the physics-guided autoencoder design introduced next.

\vspace{-1em}
\subsection{Physics-Guided Autoencoder for Structured Measurement Manifold Learning}\label{sec:AE}

For controlling how perturbations are mapped to measurement space to satisfy Definition~\ref{def:stealth-condition}, we introduce a physics-guided autoencoder that enforces such control by embedding the functional form of AC power flow into the decoder. A natural way to regulate decoder behavior is to restrict it to be linear. However, a purely linear decoder operating on the original state variables cannot represent the nonlinear AC measurement function. To reconcile these requirements, we introduce a symbolic basis transformation $f(\x)$ that lifts the system state $\x$ into a higher-dimensional feature space in which the measurement mapping becomes linear, namely
\begin{equation}\label{eq:h-rewrite}
    \h(\x) = \A f(\x),
\end{equation}
for some constant matrix $\A$. In the remainder of this subsection, we assume the existence of a symbolic basis $f(\cdot)$ that satisfies Eq.~(\ref{eq:h-rewrite}). In Section~\ref{sec:linearbasis}, we show that such a basis can be constructed directly from standard AC power flow equations. The symbolic basis explicitly encodes known nonlinear structures, including trigonometric and bilinear interactions, which shifts the burden of nonlinearity away from the decoder.

Under this lifted representation, the decoder can be implemented as a linear operator while retaining sufficient expressive power through a nonlinear encoder. This design yields predictable geometric behavior for latent perturbations and supports provable satisfaction of the acceptance condition in Definition~\ref{def:stealth-condition}. The resulting autoencoder architecture is trained by minimizing the reconstruction loss
\begin{equation}\label{eq:AE}
\min_{\thetaEnc, \thetaDec} 
\mathbb{E}_{\z}\| \z - \text{Dec}(f(\text{Enc}(\z))) \|^2.
\end{equation}

We refer to this architecture as the \emph{physics-guided autoencoder}. In contrast to a standard autoencoder whose fully nonlinear decoder is unconstrained outside the encoded image of the measurement manifold, the linear decoder in the lifted space prevents latent perturbations from being mapped to arbitrary off-manifold measurements. Theorem~\ref{thm:physics-guided-auencoder} formalizes this property by showing that the resulting perturbed measurements satisfy Definition~\ref{def:stealth-condition} under the stated assumptions.

\begin{theorem}[Manifold-Aligned Perturbations Under a physics-guided Autoencoder]
\label{thm:physics-guided-auencoder}
Assume the AC measurement model admits the lifted representation
\(
\h(\x)=\A f(\x),
\)
where $f:\mathbb{C}^{\dimx}\to\mathbb{R}^{p}$ is a known symbolic basis and
$\A\in\mathbb{R}^{\dimz\times p}$ is an unknown constant matrix.
Let the set of nominal measurements be $\mathcal{Z} = \{\z | \z=\h(\x)+\noise, \x\in\mathbb{C}^n, \noise\sim\mathcal{N}(\mathbf{0},\eR)\}$ and satisfy the detection acceptance condition
\(
\mathbb{P}_{\z\sim\mathcal{Z}}
\left(
\min_{\hat{\z}\in\mathcal{H}}\|\z-\hat{\z}\|_2^2 < \tau
\right)\ge 1-\alpha,
\)
where threshold $\tau=\chi^2_{(\dimz-\dimx), 1-\alpha}$ with significance level $\alpha$.
Consider the physics-guided autoencoder trained on samples from $\mathcal{Z}$ with encoder $\mathrm{Enc}(\cdot)$ and a linear decoder $D(\cdot)$, and suppose it achieves exact reconstruction on $\mathcal{Z}$, i.e.,
\(
D(f(\mathrm{Enc}(\z)))=\z, \forall\z\in\mathcal{Z}.
\)
Define the perturbed measurement mapping
\begin{equation}\label{eq:attack-gen}
\tilde{g}(\z)=
\z + 
D\!\left(f\!\left(\mathrm{Enc}(\z)+\mathbf{c}\right)\right)
-
D\!\left(f\!\left(\mathrm{Enc}(\z)\right)\right)
,
\end{equation}
where \(\mathbf{c}\neq \mathbf{0}\) is a latent-space perturbation.
Then the perturbed measurements satisfies Definition~\ref{def:stealth-condition}, namely
\begin{equation}
\mathbb{P}_{\z\sim\mathcal{Z}}
\left(
\min_{\hat{\z}\in\mathcal{H}}\|\tilde{g}(\z)-\hat{\z}\|_2^2 < \tau
\right)\ge 1-\alpha,
\end{equation}
where $\mathbb{P}_{\z\sim \mathcal{Z}}$ denotes the empirical probability measure over the nominal measurement set $\mathcal{Z}$.
\end{theorem}

\begin{proof}
Let $\mathcal{L}_f:=\{f(\mathrm{Enc}(\z))\mid \z\in\mathcal{Z}\}\subset\mathbb{R}^{p}$ denote the image of the nominal measurements under the encoder and symbolic basis transformation $f$. The exact reconstruction assumption indicates that $D(\mathcal{L}_f) = \mathcal{Z}$. Since $D$ is linear, its action on $\mathcal{L}_f$ is fully characterized by its action on the affine combinations of elements in $\mathcal{L}_f$, and no additional distortion is introduced beyond the learned linear structure.
Under exact reconstruction, Eq.~(\ref{eq:attack-gen}) reduces to 
\(
\tilde{g}(\z)=
\mathrm{Dec}\!\left(f\!\left(\mathrm{Enc}(\z)+\mathbf{c}\right)\right).
\) 
By construction, the output of the decoder lies in the image of $D$, and hence $\tilde{g}(\z) \in \mathcal{Z}$. 
Finally, by the dataset assumption that the nominal measurements satisfy
\(
\mathbb{P}_{\z\sim\mathcal{Z}}
(
\min_{\hat{\z}\in\mathcal{H}}\|\z-\hat{\z}\|_2^2 < \tau
)\ge 1-\alpha,
\) we obtain the conclusion 
\(
\mathbb{P}_{\z\sim\mathcal{Z}}
(
\min_{\hat{\z}\in\mathcal{H}}\|\tilde{g}(\z)-\hat{\z}\|_2^2 < \tau
)\ge 1-\alpha.
\)
\end{proof}

Theorem~\ref{thm:physics-guided-auencoder} clarifies why the lifted representation is useful for studying detectability limits. The symbolic basis $f(\cdot)$ plays two roles. First, it captures the dominant trigonometric and bilinear structure of AC measurements, which reduces the need for an unconstrained nonlinear decoder. Second, it confines latent perturbations to a structured feature space aligned with the lifted physics, which provides controlled  perturbations and preserves proximity to $\mathcal{H}$ by construction.

The existence of a structured transformation $f(\cdot)$ is natural in AC power systems because standard power flow equations contain known trigonometric and bilinear interactions among voltages and phase angles. In the next subsection, we explicitly construct a symbolic basis function $f(\cdot)$ such that the AC measurement model can be written as $\h(\x)=\A f(\x)$, which instantiates the assumption of Theorem~\ref{thm:physics-guided-auencoder}.

\vspace{-0.5em}
\subsection{Symbolic Basis for Structuring AC Measurement Manifolds}
\label{sec:linearbasis}

We construct the symbolic basis \( f(\cdot) \) by exploiting the intrinsic bilinear structure of AC power system measurements revealed by the standard Y-bus formulation. Let \( \mathbf{Y} = \mathbf{G} + j\mathbf{B} \in \mathbb{C}^{n \times n} \)
denote the network admittance matrix, where $\mathbf{G} = (g_{ik}) \in \mathbb{R}^{n\times n}$ and $\mathbf{B} = (b_{ik}) \in \mathbb{R}^{n\times n}$ are the conductance and susceptance matrices, respectively. Let the system state be the complex voltage phasor vector: 
\begin{equation}\label{eq:states}
\x =
(
V_1 e^{j\theta_1}, \cdots, V_n e^{j\theta_n}
)
\in \mathbb{C}^n,
\end{equation}
where $V_i\in\mathbb{R}$ and $\theta_i\in\mathbb{R}$ denote the voltage magnitude and phase angle at bus $i$, and $j$ denotes the imaginary unit. Then, the complex bus power injections $\mathbf{S}=\mathbf{P}+j\mathbf{Q}$ satisfy \cite{abur2004power}:
\begin{equation}\label{eq:ybus-power}
\mathbf{S}
= \mathrm{diag}(\x) \mathbf{Y}^* \x^* \Rightarrow 
\mathbf{S}_i
= \sum_{k=1}^n Y_{ik}^{\ast}\x_i\x_k^{\ast},
\end{equation}
where $\mathbf{P}$ and $\mathbf{Q}$ denote the active and reactive power injection vectors, respectively, and $(\cdot)^\ast$ denotes complex conjugation.
 Eq.~(\ref{eq:ybus-power}) implies that AC power measurements depend linearly on the set of bilinear state products $\{\x_i\x_k\}_{i,k}$, with coefficients fully determined by the network admittance matrix. This observation motivates 
 a lifted representation of the system state using these bilinear terms as coordinates, enabling a linear parameterization of the AC measurement function $\h(\cdot)$. 
 
In practice, $\h(\cdot)$ consists of standard quantities including bus active and reactive power injections, branch active and reactive power flows, and bus voltage magnitudes \cite{abur2004power}:
\begin{equation}\label{eq:h}
\boldsymbol{h}(\cdot) = \begin{bmatrix} p_{i}(\cdot) & q_{i}(\cdot) & f^{p}_{ik}(\cdot) & f^{q}_{ik}(\cdot) & V_i(\cdot) \end{bmatrix}^\top,
\end{equation}
for all buses \( i \) and branches \( (i,k) \). Here, \( p_i(\cdot) \) and \( q_i(\cdot) \) denote bus active and reactive power injection functionals, \( f^{p}_{ik}(\cdot) \) and \( f^{q}_{ik}(\cdot) \) denote branch active and reactive power flow functionals, and \( V_i(\cdot) \) denotes the bus voltage magnitude functional. The realizations of these functionals in terms of \( \x \) are:
\begin{equation}\label{eq:h-detail}
\begin{aligned}
    p_{i}(\x) &= \mathrm{Re}(\mathbf{S}_i) = \sum_{k=1}^n g_{ik}\mathrm{Re}(\x_i\x_k^{\ast}) + b_{ik}\mathrm{Im}(\x_i\x_k^{\ast}),  \\
    q_{i} (\x) &= \mathrm{Im}(\mathbf{S}_i) = \sum_{k=1}^n g_{ik}\mathrm{Im}(\x_i\x_k^{\ast}) - b_{ik}\mathrm{Re}(\x_i\x_k^{\ast}), \\
    f^{p}_{ik}(\x) &= g_{ik}^{\ell}(\mathrm{Re}(\x_i\x_i^{\ast}) - \mathrm{Re}(\x_i\x_k^{\ast})) - b_{ik}^{\ell} \mathrm{Im}(\x_i\x_k^{\ast}), \\
    f^{q}_{ik}(\x) &= -b_{ik}^{\ell}(\mathrm{Re}(\x_i\x_i^{\ast}) - \mathrm{Re}(\x_i\x_k^{\ast})) - g_{ik}^{\ell} \mathrm{Im}(\x_i\x_k^{\ast}), \\
    V_i(\x) &= V_i, 
\end{aligned}
\end{equation}
where $\mathrm{Re}(\cdot)$ and $\mathrm{Im}(\cdot)$ denote the real and imaginary parts, respectively. \(g_{ik},b_{ik}\) are Y-bus entries, and \(g^{\ell}_{ik}, b^{\ell}_{ik}\) are the series admittance parameters of branch \((i,k)\). From Eq.~(\ref{eq:h-detail}), each measurement functional is linear in the bilinear terms \(\{V_i, \mathrm{Re}(\x_i \x_k^{*}),\, \mathrm{Im}(\x_i \x_k^{*})\}_{i,k}\), with \(\mathrm{Re}(\x_i \x_i^{*})\) as the self-product case. Using
\begin{equation}\label{eq:expand}
\mathrm{Re}(\x_i \x_k^{*}) = V_i V_k \cos(\theta_{ik}), \ 
\mathrm{Im}(\x_i \x_k^{*}) = V_i V_k \sin(\theta_{ik}),
\end{equation}
where \(\theta_{ik}=\theta_i-\theta_k\), we formally define the symbolic basis
\begin{equation}\label{eq:basis}
f(\x)
=
\begin{bmatrix}
V_i \\
V_i V_k \cos(\theta_{ik}) \\
V_i V_k \sin(\theta_{ik})
\end{bmatrix},
\end{equation}
which yields a lifted representation under which the AC measurement mapping admits a linear parameterization. For example, the measurements associated with bus \( i \) satisfy:
\begin{align}\label{eq:linear-h-ij}
    \begin{bmatrix}
        p_{i}(\x)\\
        q_{i}(\x)\\
        f^{p}_{ik}(\x)\\
        f^{q}_{ik}(\x)\\
        V_{i}(\x)
    \end{bmatrix}
    =
    \underbrace{
    \begin{bmatrix}
        0 & g_{ik} & b_{ik} \\
        0 & -b_{ik} & g_{ik} \\
        0 & -\tilde{g}_{ik}^{\ell} & -b_{ik}^{\ell} \\
        0 & \tilde{b}_{ik}^{\ell} & -g_{ik}^{\ell} \\
        1 & 0 & 0
    \end{bmatrix}
    }_{A_{i}}
    \begin{bmatrix}
        V_i \\
        V_i V_k \cos(\theta_{ik}) \\
        V_i V_k \sin(\theta_{ik})
    \end{bmatrix}.
\end{align}
where $\tilde{g}_{ik}^{\ell} = \left\{\begin{array}{l}g_{ik}^{\ell}, i\neq k\\-g_{ik}^{\ell}, i=k\end{array}\right.$ and
$\tilde{b}_{ik}^{\ell} = \left\{\begin{array}{l}b_{ik}^{\ell}, i\neq k\\-b_{ik}^{\ell}, i=k\end{array}\right.$.
By stacking Eq.~(\ref{eq:linear-h-ij}) over all buses and incident branches, we obtain
the global linear representation
\(
    \h(\x) = \A f(\x),
\)
where \( \A = \text{diag}(A_1,\cdots,A_n) \in \mathbb{R}^{m \times p} \) is a block-diagonal matrix capturing system-specific coefficients. This construction explicitly realizes the structural assumption in Theorem~\ref{thm:physics-guided-auencoder}. 

The symbolic basis $f(\x)\in\mathbb{R}^p$ defined in Eq.~(\ref{eq:basis}) has dimension $p=n+2n^2 = \mathcal{O}(n^2)$ when constructed over all index pairs $(i,k)$. If connectivity information is available, the symbolic basis can be restricted to voltage interaction terms associated with adjacent buses, i.e., $k\in \mathcal{N}_{i}$, where $\mathcal{N}_{i}$ denotes the neighbor set of bus $i$. Under this sparsity-aware construction, the dimension reduces to $p=3n+2\sum_i |\mathcal{N}_i|$ which scales linearly with the number of branches and is typically much smaller than $\mathcal{O}(n^2)$ for sparse networks.

In the representation \(
    \h(\x) = \A f(\x)
\),
the matrix \( \A \) depends on unknown line parameters and is therefore unavailable to the attacker. Nevertheless, the lifted formulation exposes a linear structure of the AC measurement manifold in the coordinates $f(\x)$. This observation places the AC FDIA problem in direct analogy with model-free constructions in DC systems, where parameter matrices are unknown but the relevant measurement structure remains learnable from data \cite{li2018false,zhang2018can}. As established in Theorem~\ref{thm:physics-guided-auencoder}, this structure enables the physics-guided autoencoder to learn the measurement manifold and generate perturbed measurements that satisfy Definition~\ref{def:stealth-condition}.

% \clearpage
% ############# Experiments  #############
\section{Numerical Results}\label{sec:experiments}

This section provides numerical evidence supporting the central claim of this paper: residual-based detectors exhibit an inherent detectability limit when perturbed measurements remain sufficiently close to the measurement manifold. 
We examine whether the proposed physics-guided autoencoder can systematically generate physically consistent perturbations constrained to the measurement manifold, thereby revealing this detectability limit. Comparisons with representative model-free FDIA constructions are included to illustrate how deviations from manifold consistency translate into increased detectability.

\vspace{-0.5em}
\subsection{Experimental Setup}

We evaluate the detectability limit across diverse system topologies, including the IEEE 14-bus, 30-bus, 39-bus, 57-bus, 118-bus, and 200-bus systems. For each system, time-series measurements of $\z$ of length 1440 are generated by repeatedly solving the AC power flow equations using MATPOWER. A sensitivity analysis with respect to data volume is provided in Section~\ref{sec:exp-volume}. To enhance realism, real-world active power demand profiles from the Duquesne Light Company (Pittsburgh) are incorporated into the simulations. In addition, both load and generation are randomly scaled to introduce operating variability. Measurement noise with 2\% standard deviation is injected, consistent with typical SCADA accuracy levels.

For each system, the proposed physics-guided autoencoder is trained according to Eq.~(\ref{eq:AE}), using the symbolic basis defined in Eq.~(\ref{eq:basis}). After training, perturbed measurements $\z_a$ are generated following Eq.~(\ref{eq:attack-gen}). The perturbation vector $\mathbf{c}$ controls the direction and magnitude of the latent-space perturbation. In this study, we fix \(\mathbf{c} = (0.1, \dots, 0.1)\) across all test systems to enable consistent comparison. Under this choice of $\mathbf{c}$, the resulting impact on system states is analyzed in Section \ref{sec:exp-impact}. A sensitivity analysis of the detectability limit with respect to the perturbation magnitude is further provided in Section~\ref{sec:exp-c}.

To assess detectability under residual-based detection, we consider two representative detectors. First, we evaluate the classical chi-squared BDD, which is the most widely deployed residual-based method in practice. The residual error of $\z_a$ is calculated using Eq.~(\ref{eq:residual}), and the rate of successfully bypassing BDD is defined as
\(
    \text{succ}_{\text{BDD}} = \mathbb{P}(\J(\z_a) \le \chi^2_{(\dimz-\dimx), 1-\alpha}),
\)
where $\alpha$ is the significance level with default value $\alpha=0.05$. Second, to examine whether the observed detectability limit extends beyond BDD, we consider a learning-based residual detector \cite{musleh2022lstmAE} that flags anomalies based on reconstruction error. The rate of bypassing the learning-based detector is 
\(
     \text{succ}_{\text{learn}} = \mathbb{P}(\tilde{\J}(\z_a) \le \tau_{\text{learn}, 1-\alpha}),
\)
where $\tilde{\J}$ denotes the reconstruction error, and $\tau_{\text{learn}, 1-\alpha}$ is set to be the $(1-\alpha)$ empirical quantile of $\tilde{\J}(\z)$ on nominal data.

For reference, we include state-of-the-art model-free FDIA constructions to illustrate how deviations from manifold consistency affect detectability. Specifically, we consider AE-GAN~\cite{costilla2022attack}, which combines autoencoders with generative adversarial training, and SA-GAN~\cite{jiao2021new}, a self-attention-based GAN framework for generating realistic false measurements. These baselines are not treated as competitors, but as illustrative mechanisms that highlight the role of manifold consistency in residual-based detection.

Implementation details are as follows. In the physics-guided autoencoder model, the latent dimension is set to $d=n$, equal to the number of system states, to faithfully capture the state estimation process. A sensitivity analysis with respect to this choice is presented in Section~\ref{sec:exp-latent}. Both the encoder and decoder comprise five fully connected layers, with at most 64 neurons per layer. The autoencoder is trained for a maximum of 1000 iterations using the Adam optimizer with a learning rate of $2\times 10^{-4}$ and batch size 50. System measurements are generated using MATLAB 2022b, while all subsequent computations are executed in Python 3.12 on a Windows 10 machine equipped with an Intel Core i7 processor (2.2 GHz) and 16 GB of RAM.

\vspace{-0.5em}
\subsection{Evaluation of AC Measurement Manifold Learning} 

We first examine whether the proposed physics-guided autoencoder can faithfully reconstruct the AC measurement manifold, as required to expose the detectability limit characterized in Theorem~\ref{thm:physics-guided-auencoder}. Because the measurement space is high-dimensional and cannot be directly visualized, we consider representative low-dimensional projections using active power flow measurements from the IEEE 14-bus system. Specifically, four representative active power flow measurements between buses (2,3), (2,5), (4,7), and (6,13), are selected and denoted as pf1–pf4 in Fig.~\ref{fig:manifold-learn-correct}. Different combinations of these measurements are used to form three-dimensional projections of the measurement manifold. 
For each column in Fig.~\ref{fig:manifold-learn-correct}, the top row shows the original measurements $\z$, while the bottom row shows the reconstructed measurements $\mathrm{Dec}(f(\mathrm{Enc}(\z)))$. The close alignment between the original and reconstructed manifolds across all projections indicates that the learned representation preserves the geometric structure of the AC measurement manifold. These results empirically support the theoretical claim that encoding dominant nonlinear dependencies via the symbolic basis enables a linear decoder to reconstruct the measurement manifold, which is a key prerequisite for revealing the detectability limit of residual-based detection.

\begin{figure}[H]
\centering
\vskip -0.1in
\includegraphics[width=1\linewidth]{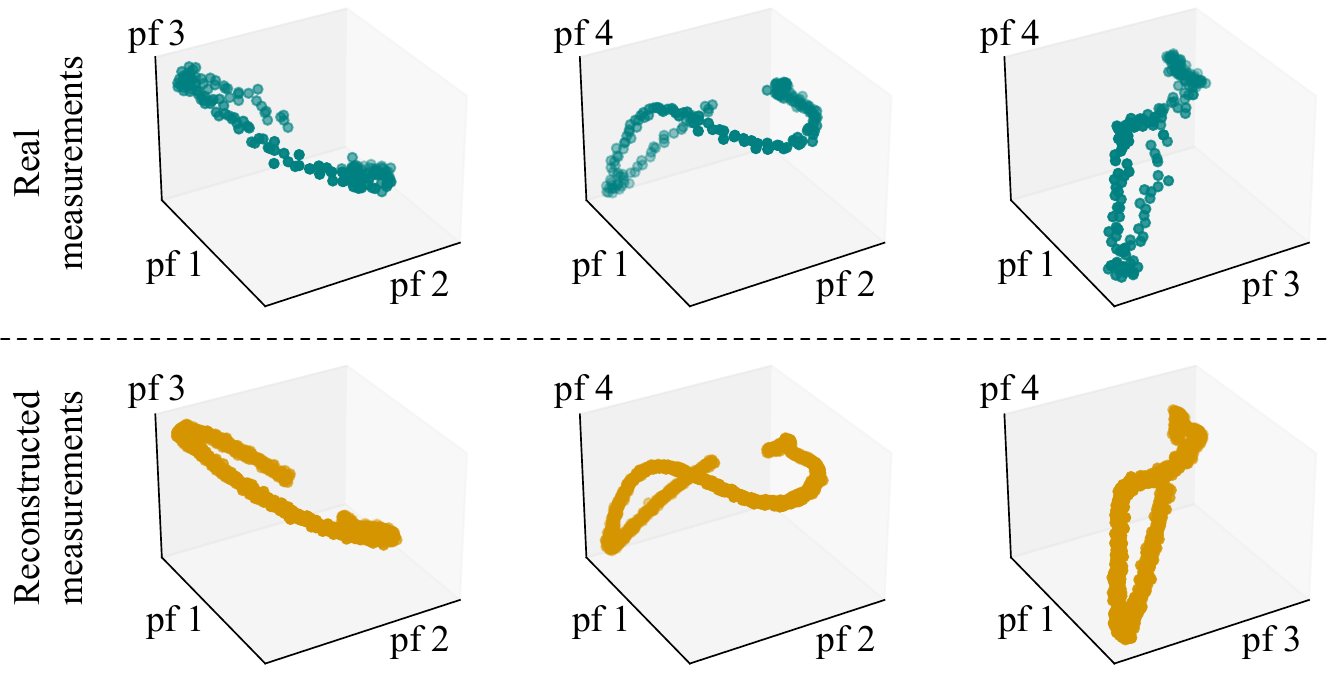}
\vskip -0.1in
\caption{Low-dimensional projections of original and reconstructed AC measurement manifolds in the IEEE 14-bus system.}
\vskip -0.1in
\label{fig:manifold-learn-correct}
\end{figure}

\vspace{-0.5em}
\subsection{Illustration of Measurement and State Perturbations}
\label{sec:exp-impact}

Building on the reconstructed measurement manifold, we next illustrate representative perturbed measurements $\z_a$ generated according to Eq.~(\ref{eq:attack-gen}). Fig.~\ref{fig:measurement-impact} compares these perturbed measurements with the original measurements in the IEEE 14-bus system using the same four representative active power flow measurements pf1–pf4. Although the perturbation introduces a noticeable magnitude shift, the temporal patterns remain closely aligned with those of the original data, indicating that the perturbations preserve the intrinsic structure of the measurements. This consistency reflects adherence to the measurement manifold and can lead to indistinguishably small residuals under residual-based detection.

\begin{figure}[h]
\centering
\vskip -0.1in
\includegraphics[width=1\linewidth]{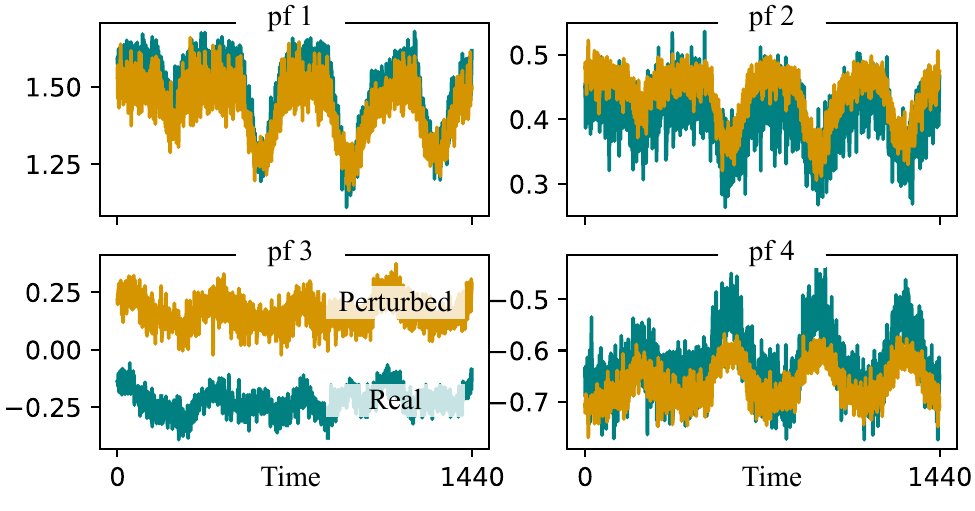}
\vskip -0.1in
\caption{Comparison between real and perturbed measurements in the IEEE 14-bus system. Perturbation vector $\mathbf{c}=(0.1,\cdots,0.1)$.}
\vskip -0.1in
\label{fig:measurement-impact}
\end{figure}

To examine the impact for state estimation, Fig.~\ref{fig:states-impact} compares the estimated system states obtained from the original measurements and the corresponding perturbed measurements $\z_a$. The first three state variables are shown as representative examples. The estimated states exhibit approximately parallel shifts that mirrors the magnitude of perturbation vector $\mathbf{c}$, indicating that the measurement-to-state mapping is largely preserved under our manifold-constrained perturbations. 

\begin{figure}[h]
\centering
\vskip -0.1in
\includegraphics[width=1\linewidth]{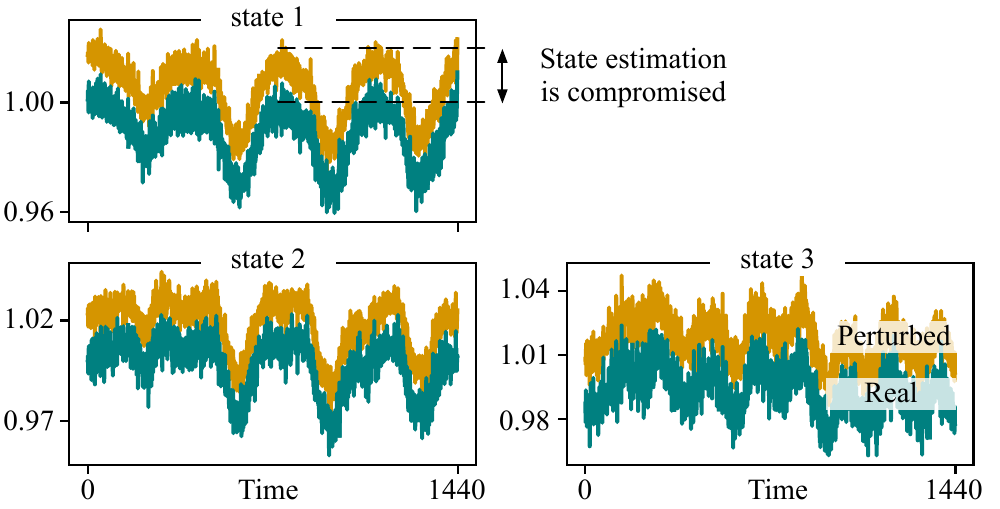}
\vskip -0.1in
\caption{Comparison of estimated system states under real and perturbed measurements in the IEEE 14-bus system. Perturbation vector $\mathbf{c}=(0.1,\cdots,0.1)$.}
\vskip -0.1in
\label{fig:states-impact}
\end{figure}

\vspace{-0.5em}
\subsection{Detectability Limits Under Bad Data Detection}
\label{sec:exp-stealthy}

To examine the detectability limit under residual-based detection, we evaluate bypass rates under two representative detectors: the classical BDD and a learning-based residual detector \cite{musleh2022lstmAE}. Fig.~\ref{fig:bypassBDD} plots the bypass rates $\text{succ}_{\text{BDD}}$ and $\text{succ}_{\text{learn}}$ as functions of the significance level \(\alpha \in [0.001, 0.1]\), comparing the proposed perturbation construction with two representative model-free baselines, AE-GAN \cite{costilla2022attack} and SA-GAN \cite{jiao2021new}. 
Across all tested systems and significance levels, perturbations generated by the proposed physics-guided autoencoder consistently achieve bypass rates exceeding the reference level $1-\alpha$ under both detectors. This behavior provides empirical support for Theorem~\ref{thm:physics-guided-auencoder} and indicates that the detectability limit is governed primarily by proximity to the AC measurement manifold rather than by the specific form of the residual-based detector.
Notably, the learning-based residual detector exhibits systematically higher bypass rates than BDD. This observation is consistent with the fact that, unlike BDD, learning-based detectors assess measurement consistency through data-driven reconstruction error and do not explicitly leverage physical measurement model. Meanwhile, the reference FDIA constructions exhibit lower BDD bypass rates, particularly at larger values of $\alpha$, indicating that residual deviations caused by imperfect manifold consistency under realistic noise and operating variability become increasingly detectable.
These observations collectively indicate that residual-based detectors, whether physics-driven (BDD) or learning-based, share a common geometric detectability limit determined by measurement-manifold consistency, which cannot be overcome by detector design alone.

\begin{figure}[h]
\centering
\vskip -0.1in
\includegraphics[width=1\linewidth]{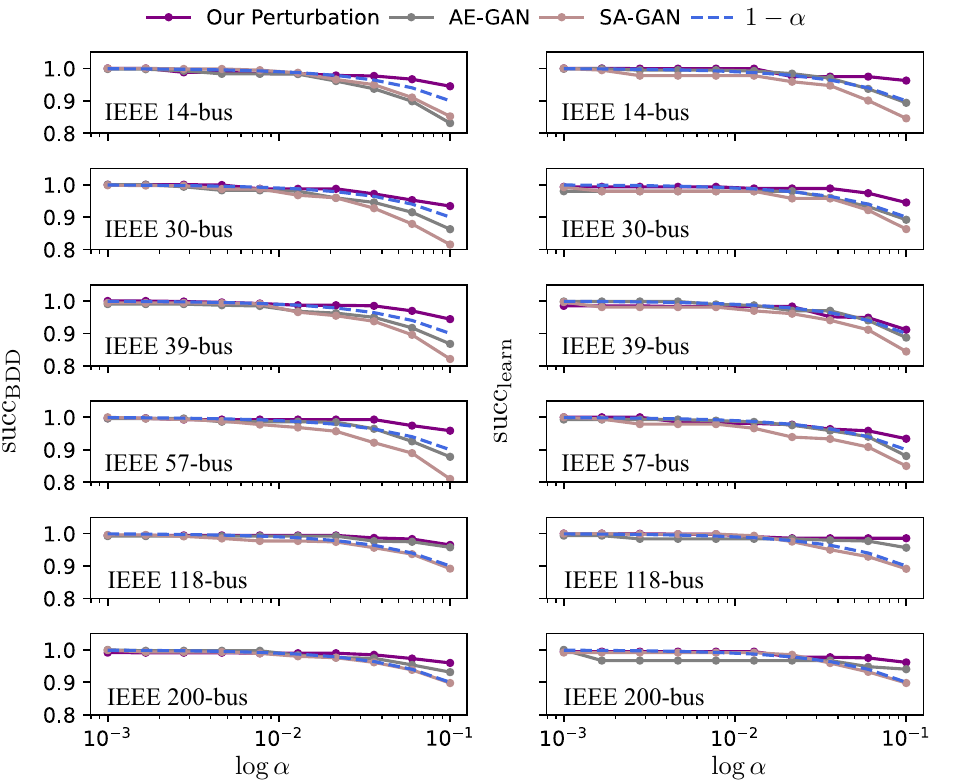}
\vskip -0.1in
\caption{Empirical bypass rates $\text{succ}_{\text{BDD}}$ and $\text{succ}_{\text{learn}}$ against various significance level \(\alpha \in [0.001, 0.1]\). Pertubation vector $\mathbf{c}=(0.1,\cdots,0.1)$.}
\vskip -0.1in
\label{fig:bypassBDD}
\end{figure}

\vspace{-0.5em}
\subsection{Ablation Study: Role of the Symbolic Basis}
A key component of the proposed perturbation construction is the symbolic basis \( f(\cdot) \), which lifts the nonlinear AC measurement model into a space where the measurement manifold admits a simpler structure. To assess its role in revealing the detectability limit, we conduct an ablation study by comparing the full physics-guided autoencoder with a variant that removes the symbolic basis \( f(\cdot) \). As shown in Fig.~\ref{fig:ablation-symbolic-basis}, omitting the symbolic basis results in larger BDD residual errors that deviate significantly from the theoretical chi-squared distribution. Consequently, the perturbative measurements are more readily detected by the BDD. 
This observation is consistent with Lemma~\ref{lemma:standard-auencoder} regarding the limitations of a standard autoencoder.
These results show that the symbolic basis is essential for systematically revealing the detectability limit by explicitly encoding the dominant nonlinear structure of the AC measurement model, rather than relying on incidental approximation accuracy from purely data-driven learning.

\begin{figure}[h]
    \centering
    \vskip -0.1in
    \includegraphics[width=1\linewidth]{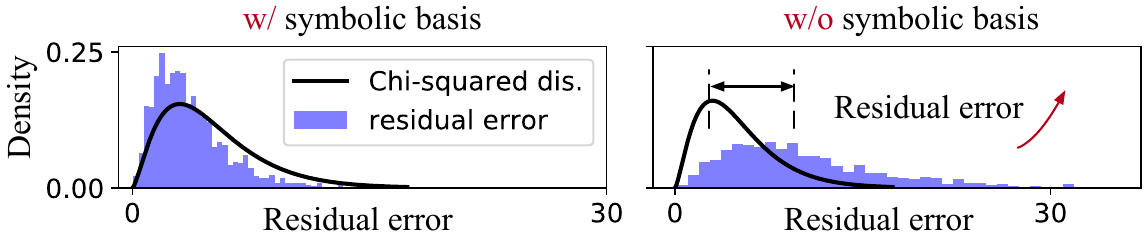}
    \vskip -0.1in
    \caption{Histogram of residual errors of perturbed measurements generated with and without the symbolism basis in the IEEE 57-bus system. Perturbation vector $\mathbf{c}=(0.1,\cdots,0.1)$.}
    \vskip -0.1in
    \label{fig:ablation-symbolic-basis}
\end{figure}

\vspace{-0.5em}
\subsection{Sensitivity Study on Latent Space Dimension}
\label{sec:exp-latent}
Another key design choice in the proposed perturbation construction is setting the autoencoder latent dimension equal to the number of system states. This choice aligns the learned representation with the state estimation process and facilitates accurate measurement reconstruction.
To evaluate the effect of latent space dimension, we conduct a sensitivity study on the IEEE 14-bus system, which has 26 free state variables, corresponding to voltage magnitudes and phase angles at 13 non-reference buses. We vary the latent dimension of autoencoder from 1 to 40, and for each setting, we report the reconstruction error of the trained autoencoder and the residual error of the generated perturbed measurements. As shown in the upper panel of Fig.~\ref{fig:errors_14bus}, the reconstruction error decreases rapidly as the latent dimension increases and becomes negligible once the latent dimension exceeds 26. In contrast, the lower panel of Fig.~\ref{fig:errors_14bus} shows that the residual error increases with the latent dimension and gradually approaches the level observed for true measurements at larger dimensions.
These results indicate that setting the latent space dimension equal to the number of system states provides a favorable trade-off, enabling accurate reconstruction of the measurement manifold while maintaining residual errors consistent with nominal operation, which is a necessary prerequisite for systematically revealing the detectability limit.

\begin{figure}[h]
\centering
\vskip -0.15in
\includegraphics[width=0.75\linewidth]{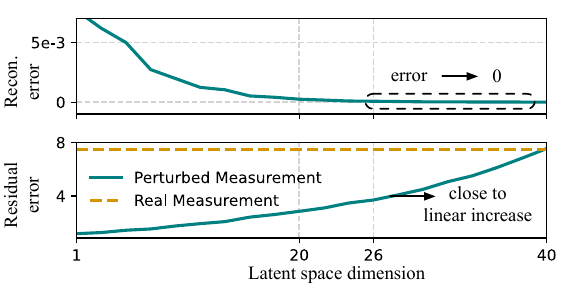}
\vskip -0.1in
\caption{Reconstruction error (upper) and residual error (lower) versus latent space dimension in the proposed physics-guided autoencoder.}
\vskip -0.1in
\label{fig:errors_14bus}
\end{figure}

\vspace{-0.5em}
\subsection{Sensitivity Study on Training Data Availability}
\label{sec:exp-volume}

The analysis in this paper assumes access to historical measurement data, consistent with prior studies on model-free FDIA analysis \cite{costilla2022attack,jiao2021new}. This assumption also reflects contemporary power system data-handling practices, where measurement streams are routinely collected, stored, and processed by external analytics platforms \cite{bidgely, verdigris}, introducing potential exposure during data transmission and storage \cite{banik2024survey}.
To examine how data availability affects the observability of the detectability limit, we conduct a sensitivity analysis with respect to training data volume. For each system, the number of training samples is gradually reduced from 1440 to 144. Fig.~\ref{fig:volume} reports the resulting BDD bypass rates for the proposed perturbation method and the baseline approaches. As the training data volume decreases, the bypass performance of all methods degrades, reflecting increasing difficulty in accurately capturing the measurement manifold. Notably, the proposed physics-guided construction exhibits a more gradual degradation compared with the baseline approaches. This behavior indicates that explicitly incorporating physical structure into manifold reconstruction improves robustness under data scarcity, enabling the detectability limit to remain observable even when historical data are limited.

\begin{figure}[h]
    \centering
    \vskip -0.1in
    \includegraphics[width=1\linewidth]{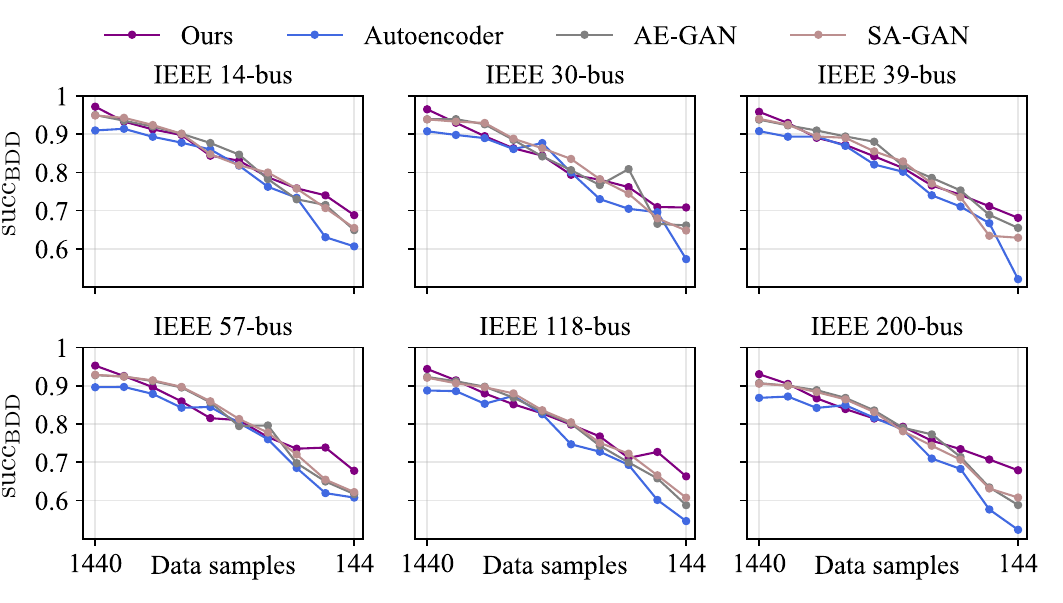}
    \vskip -0.1in
    \caption{BDD bypass rate as a function of training data volume.}
    \label{fig:volume}
    \vskip -0.1in
\end{figure}

\vspace{-0.5em}
\subsection{Sensitivity Study on Perturbation Magnitude}
\label{sec:exp-c}
We next examine how the magnitude of the latent-space perturbation affects the observability of the detectability limit. The perturbation vector \( \mathbf{c} \) characterizes the trade-off between perturbation magnitude and detectability. To study its effect, we scale the perturbation vector \( \mathbf{c} = (0.1,\cdots,0.1) \) by a scalar factor \( \gamma \in \{0.1, 0.5, 1, 2, 3\} \) and report the corresponding bypass rates of $\text{succ}_{\text{BDD}}$ and $\text{succ}_{\text{learn}}$ in Table~\ref{tab:gamma_sensitivity}.
As the perturbation magnitude increases, the BDD bypass success rate decreases from approximately 97\% to 92.4\%, indicating that larger perturbations induce greater deviations from the measurement manifold and thus increase detectability. We note that, while Theorem~\ref{thm:physics-guided-auencoder} establishes perturbation-independent detectability in the idealized setting of exact manifold reconstruction, practical factors such as finite training data and modeling imperfections introduce a dependence on perturbation magnitude. As a result, the detectability limit is more clearly observed within a moderate perturbation regime, where residual behavior remains governed primarily by proximity to the measurement manifold rather than by approximation error.

\begin{table}[h]
\centering
\vskip -0.1in
\caption{Sensitivity of bypass rates to perturbation magnitude.}
\label{tab:gamma_sensitivity}
\begin{tabular}{cccccc}
\toprule
Scaling factor $\gamma$ & 0.1 & 0.5 & 1 & 2 & 3 \\
\midrule
$\text{succ}_{\text{BDD}}$ (\%) & 97.0 & 96.3 & 95.1 & 93.6 & 92.4 \\
$\text{succ}_{\text{learn}}$ (\%) & 98.4 & 97.8 & 96.9 & 95.7 & 94.6 \\
\bottomrule
\end{tabular}
\vskip -0.1in
\end{table}

\vspace{-0.5em}
\subsection{Case Study: FDIA Under Limited Meter Manipulation}
In realistic settings, an adversary may only be able to manipulate a limited subset of measurements \(m_0 \le m\), due to physical access constraints, communication restrictions, or resource limitations. This scenario reflects practical, resource-constrained conditions in which full control over all meters is infeasible. From the perspective of detectability, this constraint raises the question of whether the measurement manifold can still be sufficiently represented when only partial measurement access is available.
Within the proposed framework, this setting is modeled by identifying a subset of \(m_0\) critical measurement dimensions that best capture the measurement manifold \(\mathcal{H}\). To this end, we extend the physics-guided autoencoder to a masked variant:
\begin{equation}\label{eq:AE_mask}
\min_{\thetaEnc, \thetaDec} 
\mathbb{E}_{\z}\left\| \z - \text{Dec}\left(f\left(\text{Enc}\left(\mask \odot \z\right)\right)\right) \right\|_{\mask}^2,
\end{equation}
where $\mask \odot \z$ denotes the dropout procedure applied to $\z$ \cite{quan2020self2self} using a binary Bernoulli mask vector $\mask$ with entries sampled independently from a Bernoulli distribution with probability \(m_0/m\). The operator $\odot$ denotes element-wise multiplication between two vectors. The norm $\left\| \cdot \right\|^2_{\mask}$ is calculated as $\left\| {\left(\boldsymbol{1} - \mask \right)} \odot \cdot \right\|^2_2$. After training convergence, we identify the subset \(\mathcal{V} \subset [m]\) corresponding to the top-\(m_0\) dimensions with the lowest reconstruction error. Using Eq. (\ref{eq:attack-gen}), the perturbed measurement is then constructed as
$
    \tilde{\z}_a = \left\{ \begin{aligned}\z_a(i) , \ i \in \mathcal{V}\\
    \z(i), \ i\not\in \mathcal{V}
    \end{aligned}\right. 
$.

We evaluate the proposed masked autoencoder on the IEEE 14-bus system using $m=14$ bus-level active power injections. The left panel of Fig. \ref{fig:meterCaseII} reports the BDD bypass rate as a function of \(m_0 \in [1, 20]\). As $m_0$ decreases, BDD bypass rate initially degrades gradually, indicating that restricting the attacker to fewer measurements increases detectability. When $m_0$ falls below a critical threshold, the bypass rate drops sharply, suggesting that a minimum subset of measurements is required to adequately represent the measurement manifold. Across all values of $m_0$, the proposed masked autoencoder consistently outperforms both a random selection strategy and a Phasor Measurement Unit (PMU) placement–based baseline \cite{baldwin2002power}, which identifies critical buses assuming full knowledge of grid parameters.
When $m_0=10$, the right panel of Fig. \ref{fig:meterCaseII} shows the identified critical buses, namely buses 1, 4, 6, 8, 10, and 14 (highlighted in red). These buses capture key structural features of the network, including major power flow junctions (buses 4 and 6), transformer-related effects (bus 8), mid-grid dynamics (bus 10), and boundary conditions (bus 14). This result illustrates that the masked autoencoder systematically identifies measurements that are critical for preserving the measurement manifold. From a defensive perspective, these locations also can correspond to potential monitoring bottlenecks where residual-based detection may be most vulnerable under partial observability.

\begin{figure}[h]
\centering
\vskip -0.1in
\includegraphics[width=1.0\linewidth]{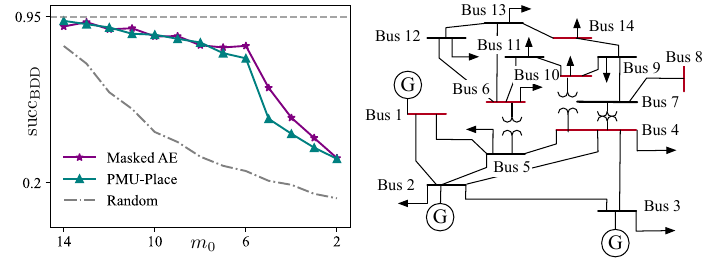}
\vskip -0.1in
\caption{Left: BDD bypass rate against $m_0\le m$ in IEEE 14-bus system. Right: important buses identified via the masked autoencoder.}
\vskip -0.1in
\label{fig:meterCaseII}
\end{figure}

\subsection{Case Study: Training with Partial Connectivity Information}

We next examine how partial grid connectivity information affects the efficiency with which the measurement manifold can be reconstructed, and consequently how the detectability limit manifests under reduced computational cost. The dominant computational burden of the proposed physics-guided autoencoder arises from the dimensionality of the symbolic basis, which scales as $p=n+2n^2 = \mathcal{O}(n^2)$ in the absence of structural information. As discussed in Section~\ref{sec:linearbasis}, when the adversary has access to network connectivity information, namely bus adjacency relations, the symbolic basis can be restricted to interactions among neighboring buses. Under this connectivity-guided construction, the symbolic basis dimension reduces to $p=3n+2\sum_i |\mathcal{N}_i|$ where $\mathcal{N}_{i}$ denotes the neighbor set of bus $i$. This dimension scales much smaller than $\mathcal{O}(n^2)$ for sparse power networks.
This scenario represents a middle ground between model-free and model-based approaches, as topology connectivity can sometimes be inferred through breaker status monitoring, substation configurations, and voltage correlation analysis. Additionally, attackers may estimate grid connectivity from historical measurement data using sparse regression, graph convolutional networks, and Bayesian networks, or retrieve topology data from public 
sources like BetterGrids or OpenGridMap.

% sources like BetterGrids \cite{BetterGrids2025} or OpenGridMap \cite{Rivera2015}.

Fig.~\ref{fig:sens1} compares the proposed perturbation construction with and without connectivity-guided symbolic basis across four metrics: (1) autoencoder training time, (2) per-sample perturbation generation time, (3) autoencoder reconstruction error, and (4) BDD bypass rate. 
The results show a substantial reduction in training time for large systems when connectivity information is incorporated, while the perturbation generation time remains nearly unchanged. This behavior is expected, as the dimension of the symbolic basis directly determines the size and training cost of the autoencoder, whereas during perturbation generation the symbolic basis is evaluated only once. Importantly, the autoencoder reconstruction error remains nearly identical between the two settings, indicating that the connectivity-guided formulation preserves the integrity of the learned measurement manifold.
Moreover, the BDD bypass rate exhibits a slight improvement when connectivity information is exploited, suggesting that constraining perturbations to a more structured and physically meaningful subspace can further sharpen the detectability limit. From a defensive perspective, these results highlight that partial exposure of connectivity information can significantly lower the barrier for systematically exploiting the detectability limits of residual-based detection.

\begin{figure}[h]
    \centering
    \vskip -0.1in
    \includegraphics[width=1\linewidth]{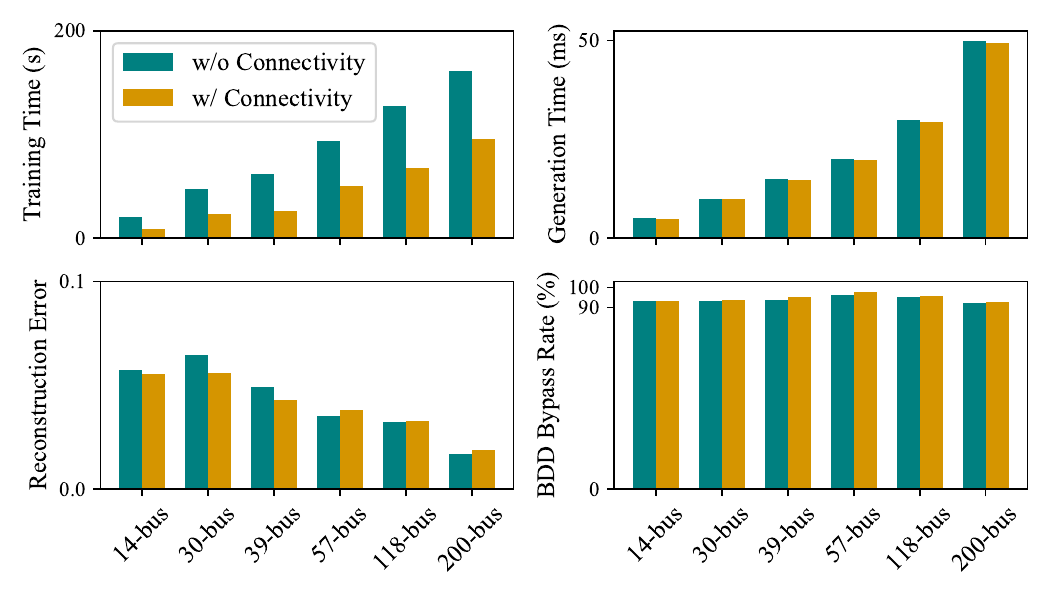}
    \vskip -0.15in
    \caption{Comparison of the proposed perturbation construction with and without connectivity-guided symbolic basis.}
    \vskip -0.1in
    \label{fig:sens1}
\end{figure}

% #############  Conclusion  #############
\vspace{-0.5em}
\section{Conclusion}\label{sec:conclusion}

This paper examined the detectability of physically consistent FDIA in AC power systems from a geometric perspective. By analyzing residual-based detection mechanisms through the lens of the AC measurement manifold, we showed that detectability is fundamentally governed by proximity to this manifold rather than by an attacker’s access to detailed network topology or line parameters. A manifold-based detectability condition was formalized to characterize this limitation. Using a physics-guided autoencoder as a constructive analytical tool, we demonstrated how manifold-aligned perturbations can systematically evade residual-based detection under realistic noise levels and data availability. Numerical studies on multiple IEEE benchmark systems confirmed the persistence of this detectability limit across different network sizes and operating conditions. Collectively, these findings expose inherent limitations of residual-based detection in AC state estimation and underscore the need for complementary detection mechanisms that extend beyond residual magnitude alone.

%%------------Bibliography--------------%%
\bibliographystyle{IEEEtran}
\bibliography{utils/bibliography.bib, utils/ref2}

@article{rudy2017data,
  title={Data-driven discovery of partial differential equations},
  author={Rudy, Samuel H and Brunton, Steven L and Proctor, Joshua L and Kutz, J Nathan},
  journal={Science advances},
  volume={3},
  number={4},
  pages={e1602614},
  year={2017},
  publisher={American Association for the Advancement of Science}
}

@article{musleh2022attack,
  title={Attack detection in automatic generation control systems using LSTM-based stacked autoencoders},
  author={Musleh, Ahmed S and Chen, Guo and Dong, Zhao Yang and Wang, Chen and Chen, Shiping},
  journal={IEEE Transactions on Industrial Informatics},
  volume={19},
  number={1},
  pages={153--165},
  year={2022},
  publisher={IEEE}
}

@article{paul2022modified,
  title={Modified grey wolf optimization approach for power system transmission line congestion management based on the influence of solar photovoltaic system},
  author={Paul, Kaushik},
  journal={International Journal of Energy and Environmental Engineering},
  volume={13},
  number={2},
  pages={751--767},
  year={2022},
  publisher={Springer}
}

@article{yang2022blind,
  title={Blind false data injection attacks against state estimation based on matrix reconstruction},
  author={Yang, Haosen and He, Xing and Wang, Ziqiang and Qiu, Robert C and Ai, Qian},
  journal={IEEE Transactions on Smart Grid},
  volume={13},
  number={4},
  pages={3174--3187},
  year={2022},
  publisher={IEEE}
}

@inproceedings{rahman2014impact,
  title={Impact analysis of topology poisoning attacks on economic operation of the smart power grid},
  author={Rahman, Mohammad Ashiqur and Al-Shaer, Ehab and Kavasseri, Rajesh},
  booktitle={IEEE 34th International Conference on Distributed Computing Systems},
  pages={649--659},
  year={2014}
}

@book{abur2004power,
title={Power system state estimation: theory and implementation},
author={Abur, Ali and Exposito, Antonio Gomez},
year={2004},
publisher={CRC press}
}

@article{liu2011false,
  title={False data injection attacks against state estimation in electric power grids},
  author={Liu, Yao and Ning, Peng and Reiter, Michael K},
  journal={ACM Transactions on Information and System Security},
  volume={14},
  number={1},
  pages={1--33},
  year={2011},
  publisher={ACM New York, NY, USA}
}

@article{jin2018power,
  title={Power grid AC-based state estimation: Vulnerability analysis against cyber attacks},
  author={Jin, Ming and Lavaei, Javad and Johansson, Karl Henrik},
  journal={IEEE Transactions on Automatic Control},
  volume={64},
  number={5},
  pages={1784--1799},
  year={2018},
  publisher={IEEE}
}

@article{li2018false,
  title={False data injection attacks with incomplete network topology information in smart grid},
  author={Li, Yuancheng and Wang, Yuanyuan},
  journal={IEEE Access},
  volume={7},
  pages={3656--3664},
  year={2018},
  publisher={IEEE}
}

@article{zhang2018can,
  title={Can attackers with limited information exploit historical data to mount successful false data injection attacks on power systems?},
  author={Zhang, Jiazi and Chu, Zhigang and Sankar, Lalitha and Kosut, Oliver},
  journal={IEEE Transactions on Power Systems},
  volume={33},
  number={5},
  pages={4775--4786},
  year={2018},
  publisher={IEEE}
}

@article{yu2015blind,
  title={Blind false data injection attack using PCA approximation method in smart grid},
  author={Yu, Zong-Han and Chin, Wen-Long},
  journal={IEEE Transactions on Smart Grid},
  volume={6},
  number={3},
  pages={1219--1226},
  year={2015},
  publisher={IEEE}
}

@article{esmalifalak2015stealthy,
  title={A stealthy attack against electricity market using independent component analysis},
  author={Esmalifalak, Mohammad and Nguyen, Huy and Zheng, Rong and Xie, Le and Song, Lingyang and Han, Zhu},
  journal={IEEE Systems Journal},
  volume={12},
  number={1},
  pages={297--307},
  year={2015},
  publisher={IEEE}
}

@article{costilla2022attack,
  title={Attack Power System State Estimation by Implicitly Learning the Underlying Models},
  author={Costilla-Enriquez, Napoleon and Weng, Yang},
  journal={IEEE Transactions on Smart Grid},
  volume={14},
  number={1},
  pages={649--662},
  year={2022},
  publisher={IEEE}
}

@ARTICLE{hug2012vulnerability,
author={G. {Hug} and J. A. {Giampapa}},
journal={IEEE Transactions on Smart Grid}, 
title={Vulnerability Assessment of AC State Estimation With Respect to False Data Injection Cyber-Attacks}, 
year={2012},
volume={3},
number={3},
pages={1362-1370},
doi={10.1109/TSG.2012.2195338}
}

@ARTICLE{wang2020detection,
author={Z. {Wang} and H. {He} and Z. {Wan} and Y. {Sun}},
journal={IEEE Transactions on Smart Grid}, 
title={Detection of False Data Injection Attacks in AC State Estimation Using Phasor Measurements}, 
year={2020},
volume={},
number={},
pages={1-1},
doi={10.1109/TSG.2020.2972781}}

@article{costilla2020combining,
  title={Combining newton-raphson and stochastic gradient descent for power flow analysis},
  author={Costilla-Enriquez, Napoleon and Weng, Yang and Zhang, Baosen},
  journal={IEEE Transactions on Power Systems},
  volume={36},
  number={1},
  pages={514--517},
  year={2020},
  publisher={IEEE}
}

@book{goodfellow2016deep,
  title={Deep Learning},
  author={Goodfellow, Ian and Bengio, Yoshua and Courville, Aaron},
  year={2016},
  publisher={The MIT Press}
}

@inproceedings{quan2020self2self,
  title={Self2self with dropout: Learning self-supervised denoising from single image},
  author={Quan, Yuhui and Chen, Mingqin and Pang, Tongyao and Ji, Hui},
  booktitle={Proceedings of the IEEE/CVF conference on computer vision and pattern recognition},
  pages={1890--1898},
  year={2020}
}

@misc{economist2010cyber,
author = {{The Economist}},
title = {{A cyber-missile aimed at Iran?}},
url = {https://www.economist.com/babbage/2010/09/24/a-cyber-missile-aimed-at-iran},
year={2010}
}

@article{baldi1989neural,
  title={Neural networks and principal component analysis: Learning from examples without local minima},
  author={Baldi, Pierre and Hornik, Kurt},
  journal={Neural networks},
  volume={2},
  number={1},
  pages={53--58},
  year={1989},
  publisher={Elsevier}
}

@ARTICLE{zimmerman2010matpower,
author={R. D. {Zimmerman} and C. E. {Murillo-Sánchez} and R. J. {Thomas}},
journal={IEEE Transactions on Power Systems}, 
title={MATPOWER: Steady-State Operations, Planning, and Analysis Tools for Power Systems Research and Education}, 
year={2011},
volume={26},
number={1},
pages={12-19},
doi={10.1109/TPWRS.2010.2051168}}

@article{zhang2018limited,
  title={Can attackers with limited information exploit historical data to mount successful false data injection attacks on power systems?},
  author={Zhang, J. and Chu, Z. and Sankar, L. and Kosut, O.},
  journal={IEEE Transactions on Power Systems},
  year={2018}
}

@incollection{wang2019cpad,
  title={Cyber-physical anomaly detection for power grid with machine learning},
  author={Wang, Pengyuan and Govindarasu, Manimaran},
  booktitle={Industrial Control Systems Security and Resiliency: Practice and Theory},
  pages={31--49},
  year={2019},
  publisher={Springer}
}

@article{musleh2022lstmAE,
  title={Attack detection in automatic generation control systems using LSTM-based stacked autoencoders},
  author={Musleh, Ahmed S and Chen, Guo and Dong, Zhao Yang and Wang, Chen and Chen, Shiping},
  journal={IEEE Transactions on Industrial Informatics},
  volume={19},
  number={1},
  pages={153--165},
  year={2022},
  publisher={IEEE}
}

@inproceedings{Zideh2023PIConvAE,
  title={Physics-informed convolutional autoencoder for cyber anomaly detection in power distribution grids},
  author={Zideh, Mehdi Jabbari and Solanki, Sarika Khushalani},
  booktitle={IEEE Power \& Energy Society General Meeting},
  pages={1--5},
  year={2024}
}

@article{liu2025survey,
  title={A survey on diffusion models for anomaly detection},
  author={Liu, Jing and Ma, Zhenchao and Wang, Zepu and Zou, Chenxuanyin and Ren, Jiayang and Wang, Zehua and Song, Liang and Hu, Bo and Liu, Yang and Leung, Victor},
  journal={arXiv preprint arXiv:2501.11430},
  year={2025}
}

@article{mantravadi2020securing,
  title={Securing IT/OT links for low power IIoT devices: design considerations for industry 4.0},
  author={Mantravadi, Soujanya and Schnyder, Reto and M{\o}ller, Charles and Brunoe, Thomas Ditlev},
  journal={IEEE Access},
  volume={8},
  pages={200305--200321},
  year={2020},
  publisher={IEEE}
}

@ARTICLE{8686241,
  author={Soltan, Saleh and Mittal, Prateek and Poor, H. Vincent},
  journal={IEEE Transactions on Power Systems}, 
  title={Line Failure Detection After a Cyber-Physical Attack on the Grid Using Bayesian Regression}, 
  year={2019},
  volume={34},
  number={5},
  pages={3758-3768},
  doi={10.1109/TPWRS.2019.2910396}}

@ARTICLE{8468098,
  author={Che, Liang and Liu, Xuan and Li, Zuyi and Wen, Yunfeng},
  journal={IEEE Transactions on Power Systems}, 
  title={False Data Injection Attacks Induced Sequential Outages in Power Systems}, 
  year={2019},
  volume={34},
  number={2},
  pages={1513-1523},
  doi={10.1109/TPWRS.2018.2871345}}

@article{jiao2021new,
  title={A new AC false data injection attack method without network information},
  author={Jiao, Runhai and Xun, Gangyi and Liu, Xuan and Yan, Guangwei},
  journal={IEEE Transactions on Smart Grid},
  volume={12},
  number={6},
  pages={5280--5289},
  year={2021},
  publisher={IEEE}
}

@article{yang2023false,
  title={A False Data Injection Attack Approach Without Knowledge of System Parameters Considering Measurement Noise},
  author={Yang, Haosen and Wang, Ziqiang},
  journal={IEEE Internet of Things Journal},
  year={2023},
  publisher={IEEE}
}

@article{stouffer2022guide,
  title={Guide to operational technology (ot) security},
  author={Stouffer, Keith and Pease, Michael and Tang, C and Zimmerman, Timothy and Pillitteri, Victoria and Lightman, Suzanne},
  journal={National Institute of Standards and Technology: Gaithersburg, MD, USA},
  year={2022}
}

@article{di2018triton,
  title={TRITON: The first ICS cyber attack on safety instrument systems},
  author={Di Pinto, Alessandro and Dragoni, Younes and Carcano, Andrea},
  journal={Proc. Black Hat USA},
  volume={2018},
  pages={1--26},
  year={2018}
}

@article{wang2019cyber,
  title={Cyber-physical anomaly detection for power grid with machine learning},
  author={Wang, Pengyuan and Govindarasu, Manimaran},
  journal={Industrial Control Systems Security and Resiliency: Practice and Theory},
  pages={31--49},
  year={2019},
  publisher={Springer}
}

@misc{bidgely,
  author = {{Bidgely Inc.}},
  title = {Energy disaggregation services},
  howpublished = {\url{https://www.bidgely.com}},
  note = {Accessed: 2025-07-02}
}

@misc{verdigris,
  author = {{Verdigris Technologies}},
  title = {AI-powered energy monitoring},
  howpublished = {\url{https://verdigris.co}},
  note = {Accessed: 2025-07-02}
}

@article{banik2024survey,
  title={Survey on Vulnerability Testing in the Smart Grid},
  author={Banik, Shampa and Rogers, Michael and Mahajan, Satish M and Emeghara, Chikezie M and Banik, Trapa and Craven, Robert},
  journal={IEEE Access},
  year={2024},
  publisher={IEEE}
}

@ARTICLE{xiao2025,
  author={Xiao, Chenhan and Costilla-Enriquez, Napoleon and Weng, Yang},
  journal={IEEE Open Access Journal of Power and Energy}, 
  title={Guaranteed False Data Injection Attack Without Physical Model}, 
  year={2025},
  volume={12},
  number={},
  pages={429-441},
}

@article{baldwin2002power,
  title={Power system observability with minimal phasor measurement placement},
  author={Baldwin, Thomas L and Mili, Lamine and Boisen, Monte B and Adapa, Ram},
  journal={IEEE Transactions on Power systems},
  volume={8},
  number={2},
  pages={707--715},
  year={2002},
  publisher={IEEE}
}

\end{document}